\documentclass{amsproc}
\usepackage{amssymb}
\usepackage{latexsym} 
\usepackage{amsmath}
\usepackage{caption}

\usepackage{pstricks}
\usepackage{diagrams}
\usepackage{color}
\usepackage{epsfig}

\newtheorem{theorem}{Theorem}[section]
\newtheorem{lemma}[theorem]{Lemma}
\newtheorem{proposition}[theorem]{Proposition}
\newtheorem{corollary}[theorem]{Corollary}

\theoremstyle{definition}
\newtheorem{definition}[theorem]{Definition}
\newtheorem{example}[theorem]{Example}

\theoremstyle{remark}
\newtheorem{remark}[theorem]{Remark}

\numberwithin{equation}{section}

\def\bR{\begin{color}{red}} 
\def\bB{\begin{color}{blue}}
\def\bM{\begin{color}{magenta}}
\def\bC{\begin{color}{cyan}}
\def\bW{\begin{color}{white}}
\def\bBl{\begin{color}{black}}
\def\bG{\begin{color}{green}}
\def\bY{\begin{color}{yellow}}
\def\e{\end{color}}

\newcommand{\ket}[1]{|#1\rangle}

\newcommand{\one}{{\rm I}}
\newcommand{\two}{{\rm I}\hspace{-0.4mm}{\rm I}}

\newcommand{\four}{{\rm I}\hspace{-0.4mm}{\rm V}}

\begin{document}

\title[Spekkens's toy theory as a category]{Spekkens's toy theory as a category of processes}


\author{Bob Coecke}
\address{Oxford University, Department of Computer Science}
%\curraddr{Oxford University, Department of Computer Science}
\email{coecke@cs.ox.ac.uk}
\thanks{This work was supported by EPSRC Advanced Research Fellowship EP/D072786/1 and by EU-FP6-FET STREP QICS}

\author{Bill Edwards}
\address{Oxford University, Department of Computer Science} 
\curraddr{Perimeter Institute for Theoretical Physics}
\email{wae28@hotmail.com}
\thanks{This work was supported by EPSRC PhD Plus funding at Oxford University Computing Laboratory and an FQXi Large Grant.}


\begin{abstract}
We provide two mathematical descriptions of Spekkens's toy qubit theory, an inductively one in terms of a small set of generators, as well as an explicit closed form description. It is a subcategory \textbf{MSpek} of the category of finite sets, relations and the cartesian product. States of maximal knowledge form a subcategory \textbf{Spek}. This establishes the consistency of the toy theory, which has previously only been constructed for at most four systems.  Our model also shows that the theory is closed under both parallel and sequential composition of operations (= symmetric monoidal structure), that it obeys map-state duality (= compact closure), and that states and effects are in bijective correspondence (= dagger structure). From the perspective of categorical quantum mechanics, this provides an interesting alternative model which enables us to describe many quantum phenomena in a discrete manner, and to which mathematical concepts such as basis structures, and complementarity thereof, still apply.
Hence, the framework of categorical quantum mechanics has delivered on its promise to encompass theories other than quantum theory.
\end{abstract}

\maketitle

%%%%%%%%%%%%%%%%%%%%%%%%%
%%%%%%%%%%%%%%%%%%%%%%%%%
\section{Introduction}
%%%%%%%%%%%%%%%%%%%%%%%%%
%%%%%%%%%%%%%%%%%%%%%%%%%

In 2007 Rob Spekkens proposed a toy theory \cite{Spekkens} with the aim of showing that many of the characteristic features of quantum mechanics could result from a restriction on our knowledge of the state of an essentially classical system. The theory describes a simple type of system which mimics many of the features of a quantum qubit.  The success of the toy theory in replicating characteristic quantum behaviour is, in one sense, quite puzzling, since the mathematical structures employed by the two theories are quite different. Quantum mechanics represents states of systems by vectors in a Hilbert space, while processes undergone by systems are represented by linear maps. In contrast, the toy theory represents states by subsets and processes by relations. This `incomparability' means that the mathematical origins of the similarities (and differences) between the two theories are not easy to pinpoint. 

In this paper we consider the toy theory from a new perspective - by looking at the `algebra' of how the processes of the theory combine. The mathematical structure formed by the composition of processes in any physical theory is a \emph{symmetric monoidal category} \cite{MacLane}, with the objects of the category representing systems, and the morphisms representing processes undergone by these systems - we term this the \emph{process category} of the theory. Work initiated by Abramsky and Coecke \cite{CatSem}, and continued by many other authors \cite{SellDCC, RedGreen, wosums} has investigated the structure of the process category of quantum mechanics. Mathematical structures of this category have been identified which correspond directly to key physical features of the theory. In principle, any theory whose processes form a category with these mathematical features should exhibit these quantum-like physical features. It will be gratifying then to note that the process category of the toy theory shares many of these features, thus `accounting for' the similarities which this theory has with quantum mechanics. 

The main aim of this paper is to characterise the process category of the toy theory. This turns out to be less straightforward than one might have hoped! In the case of quantum mechanics it is straightforward to see that the states of a system correspond exactly with the vectors of the Hilbert space describing that system, and that processes correspond exactly with linear maps; thus we can immediately conclude that the process category of quantum mechanics is \textbf{Hilb}, whose objects are Hilbert spaces and whose morphisms are linear maps. Such a quick and easy statement of the process category for the toy theory is not possible. This is due largely to the way that the valid states and processes of the theory are defined: this is via an inductive procedure where the valid states and transformations for a collection of $n-1$ systems must be known before those for a collection of $n$ systems can be deduced. 

There is also a certain degree of ambiguity in the way in which the theory was originally stated, and, since the toy theory was never fully constructed (until now), its consistency was not even guaranteed.   Another feature which was not addressed in the theory is \em compositional closure\em, that is: 
\em If we compose two valid processes, either in parallel or sequentially (when types match), do we again obtain a valid process? \em 
Obviously, this is a natural operational requirement, and as indicated above, this notion of compositional closure is the operational cornerstone to modelling physical processes in a  symmetric monoidal category. In this paper, it is exactly the compositional closure assumption which allows us to formulate the toy theory in terms of no more than a few generators, with clear operational meanings.

The structure of the paper is as follows. We begin with a very brief summary of the toy bit theory, and a discussion of the ambiguities in its definition. We then define a categories \textbf{Spek} and  \textbf{MSpek} which we claim is the process category for the toy theory. This is a sub-category of \textbf{Rel}, whose objects are sets and whose morphisms are relations. The next section is devoted to showing the general form of the relations which constitute the morphisms of \textbf{MSpek}. We then go on to argue that these relations are exactly those which describe the processes of the toy theory, thus demonstrating that \textbf{MSpek} is the process category of the toy theory. Finally we note some of the key categorical features of \textbf{MSpek} and link these to the characteristically quantum behaviour exhibited by the toy bit theory.

The category \textbf{Spek} was first proposed by us in \cite{CE}, and we provide some explicit pointers to useful information therein throughout this text.

%%%%%%%%%%%%%%%%%%%%%%%%%
%%%%%%%%%%%%%%%%%%%%%%%%%
\section{The toy theory}
%%%%%%%%%%%%%%%%%%%%%%%%%
%%%%%%%%%%%%%%%%%%%%%%%%%

For full details of Spekkens's toy bit theory, the reader is referred to the original paper \cite{Spekkens}. Here we provide a very brief summary of the key points. There is just one type of elementary system in the theory, which can exist in one of four states. We will denote the set of these four states, the \emph{ontic state space} by $\four = \{1,2,3,4\}$. Alternatively we can depict it graphically as:
\begin{equation}
\def\JPicScale{0.4}
\input{onticstatespace.pst}
\end{equation}
The ontic states are to be distinguished from the \emph{epistemic} states, which describe the extent of our knowledge about which ontic state the system occupies. For example we might know that the system is either in ontic state 1 or 2. We would depict such an epistemic state in the following fashion: 
\begin{equation}
\def\JPicScale{0.4}
\input{epistemicexample.pst}
\end{equation}
The epistemic state is intended to be the analogue of the quantum state. Note that an epistemic state is simply a subset of the ontic state space. Given a composite system of $n$ elementary systems, its ontic state space is simply the Cartesian product of the ontic state spaces of the composite systems, $\four^n$. Thus such a system has $4^n$ ontic states.  

The key premise of the theory is that our knowledge of the ontic state is restricted in a specific way so that only certain epistemic states are allowed. We refer the reader to the original paper for the full statement of this epistemic restriction, which Spekkens refers to as the \emph{knowledge balance principle}. It suffices here to say that, given an $n$-component system, an epistemic state is a $2^n, 2^{n+1},\dots$ or $2^{2n}$ element subset of the ontic state space, the $2^n$ case being the situation of maximal knowledge and the $2^{2n}$ case being the situation of total ignorance.  

For a single system this means that there are six epistemic states of maximal knowledge: 
\begin{equation}
\def\JPicScale{0.4}
\input{1epistemic.pst}
\end{equation}
and one of non-maximal knowledge:
\begin{equation}
\def\JPicScale{0.4}
\input{1epistemicmixed.pst}
\end{equation}

A consequence of the knowledge balance principle is that any measurement on a system inevitably results in a probabilistic disturbance of that system's ontic state. A measurement in the toy theory essentially corresponds to asking which of a collection of subsets of the ontic state space contains the ontic state. An example would be a measurement on a single system corresponding to the question "is the ontic state in the subset $\{1,2\}$ or the subset $\{3,4\}$?". If the initial epistemic state of the system was the subset $\{1,3\}$ then answering this question would allow us to pin down the ontic state precisely, to 1 if we get the first answer or 3 if we get the second, in violation of the knowledge balance principle. Thus we hypothesise that if, for example, we get the first answer to our question, then the ontic state undergoes a random disturbance, either remaining in state 1 or moving into state 2, with equal probability. The epistemic state following the measurement then would be $\{1,2\}$. The effect of such a probabilistic disturbance on the \emph{epistemic} state is best modelled by viewing it as a \emph{relation} on the ontic state space, defined by $1\sim\{1,2\}$, $2\sim\{1,2\}$, $3\sim\emptyset$, $4\sim\emptyset$, where the notation denotes that the element 1 in the domain relates to both the elements 1 and 2 in the codomain, the element 2 in the domain, while the elements 3 and 4 in the domain relate to no elements in the codomain. Below we will omit specification of those elements in the domain that relate to no elements in the codomain. Thus in general the transformations of the theory are described by \emph{relations}. 

In his original paper Spekkens derives the states and transformations allowed in the theory, apparently appealing only to the following three principles:
\begin{enumerate}
\item The epistemic state of any system must satisfy the knowledge balance principle globally i.e. it should be a $2^n, 2^{n+1},\dots$ or $2^{2n}$ element subset of the ontic state space. 
\item When considering a composite system, the `marginal' epistemic state of any subsystem should also satisfy the knowledge balance principle. By the marginal epistemic state we mean the following. Suppose the whole system has $n$ components, and we are interested in an $m$-component subsystem. The epistemic state of the whole system will be some set of $n$-tuples: to get the marginal state on the $m$-component system we simply delete the $n-m$ entries from each tuple which correspond to the subsystems which are not of interest to us. 
\item Applying a valid transformation to a valid state should result in a valid state. 
\end{enumerate}

The second principle is used, for example, to rule this out as an epistemic state for a two-component system: 
\begin{equation}
\def\JPicScale{0.4}
\input{ruledoutby2.pst}
\end{equation}
The third principle is used extensively. It is invoked to show that the transformations on a single element system constitute permutations on the set of ontic states, since any other function on $\four$ would lead to some valid epistemic state being transformed into an invalid state. 

Another illuminating example of the third principle in action is the elimination of this state as a valid epistemic state: 
\begin{equation}\label{ruledoutby3}
\def\JPicScale{0.4}
\input{ruledoutby3.pst}
\end{equation}
This state would be allowed by the first and second principles, but upon making a certain measurement on one of the systems, the resulting measurement disturbance would transform it into this state: 
\begin{equation}
\def\JPicScale{0.4}
\input{ruledoutby1.pst}
\end{equation}
which clearly fails to satisfy the first principle. 

The mode of definition of the theory, via the three principles stated earlier, raises some interesting issues. Firstly, this approach seems to be necessarily iterative. Compare it, for example, to how we would describe the form of valid states in quantum mechanics - in one line we can say that they are the normalised vectors of the system's state space. Secondly, do these rules actually uniquely define the toy theory? There does seem to be a problem with the third rule. When we use it to rule out the state in diagram \ref{ruledoutby3}, our argument is that we already know what the valid measurement disturbance transformations on a single elementary system are, and when we apply one of them to this state we obtain a state which clearly violates the knowledge balance principle. However, it is not clear that we could not have made the alternative choice - that this state should be valid, and therefore that the transformation which we had previously thought was valid could no longer be considered as such. It seems that considerations other than the three rules above come into deciding which should be the valid states of the theory, but it is nowhere clearly stated exactly what they are. 

The second point touched slightly on our final issue: is the theory as Spekkens presents it consistent? He derives valid states and transformations for systems with up to three elementary components. However, can we be sure that these states and transformations, when combined in more complex situations involving four or more elementary systems, won't yield a state which clearly violates the knowledge balance principle? Currently there seems to be no such proof of consistency. In fact, in the process of re-expressing the theory in categorical terms, we will develop such a proof. 

%%%%%%%%%%%%%%%%%%%%%%%%%
%%%%%%%%%%%%%%%%%%%%%%%%%
\section{Process categories of quantum-like theories}\label{secCategories}
%%%%%%%%%%%%%%%%%%%%%%%%%
%%%%%%%%%%%%%%%%%%%%%%%%%

We briefly review the key structural features of the process categories of quantum-like theories.  A physicist-friendly tutorial of the category theoretic preliminaries is in \cite{CatsII}.  A survey on the recent applications of these to quantum theory and quantum information is in \cite{ContPhys}.

\begin{definition}
\label{SMCdef}
A \emph{symmetric monoidal category} $(\mathcal{C}, I, -\otimes -)$ is a category equipped with the following extra structure: a bifunctor $-\otimes - : \mathcal{C} \times \mathcal{C} \rightarrow \mathcal{C}$; a \emph{unit object} $I$; and four natural isomorphisms, left and right unit: $\lambda_{A}: A \cong I\otimes A$, $\rho_{A}: A \cong A\otimes I$, associative: $\alpha_{A,B,C}: (A\otimes B)\otimes C \cong A\otimes (B\otimes C)$ and commutative: $\sigma_{A,B}: A\otimes B \cong B\otimes A$. 
Furthermore these objects and natural isomorphisms obey a series of \emph{coherence conditions} \cite{MacLane}.
\end{definition}
As argued in \cite{CatsII}, the process category of any physical theory is a SMC. The bifunctor $-\otimes -$ is interpreted as adjoining two systems to make a larger compound system. \textbf{Hilb} becomes a SMC with the tensor product as the bifunctor, and $\mathbb{C}$ as the unit object. \textbf{Rel} is a SMC with the Cartesian product as the bifunctor and the single element set as the unit object. 

There is a very useful graphical language for describing SMCs, due to Joyal and Street \cite{JSSMC}, which we will make extensive use of. It traces back to Penrose's earlier diagrammatic notation for abstract tensors  \cite{Penrose}.  

In this language we represent a morphism $f:A\rightarrow B$ by a box:
\begin{equation}
\def\JPicScale{0.4}
\input{genericmorphism.pst}
\end{equation}
$g\circ f$, the composition of morphisms $f:A \rightarrow B$ and $g:B\rightarrow C$ is depicted as:
\begin{equation}
\def\JPicScale{0.4}
\input{composingmorphisms.pst}
\end{equation}
The identity morphism $1_{A}$ is actually just written as a straight line --- this makes sense if you imagine composing it with another morphism.
Turning to the symmetric monoidal structure, a morphism $f: A\otimes B \rightarrow C\otimes D$ is depicted: 
\begin{equation}
\def\JPicScale{0.4}
\input{genericSMCmorphism.pst}
\end{equation}
and if $f:A\rightarrow B$ and $g:C \rightarrow D$ then $f \otimes g$ is depicted as:
\begin{equation}
\def\JPicScale{0.4}
\input{ftensorg.pst}
\end{equation}
The identity object $I$ is not actually depicted in the graphical language. Morphisms $\psi: I\rightarrow A$ and $\pi: A\rightarrow I$ are written as:
 \begin{equation}
 \def\JPicScale{0.4}
 \input{prepanddest.pst}
 \end{equation}
The associativity and left and right unit natural isomorphisms are also implicit in the language. The symmetry natural isomorphism is depicted as: 
 \begin{equation}
 \def\JPicScale{0.4}
 \input{symmetryiso.pst}
 \end{equation}
In fact the graphical language is more than just a useful tool; it enables one to derive all equational statements that follow from the axioms of a SMC: 
\begin{theorem}[\cite{JSSMC,SelingerSurvey}]
Two morphisms in a symmetric monoidal category can be shown to be equal using the axioms of a  SMC iff the diagrams corresponding to these morphisms in the graphical language are isomorphic, where by diagram isomorphism we mean that the boxes and wires of the first are in bijective correspondence with the second, preserving the connections between boxes and wires. 
\end{theorem}

The categories described by the graphical language are in fact \em strict SMCs, that is, those for which $\lambda_{A}$, $\rho_{A}$ and $\alpha_{A,B,C}$ are equalities.  Mac Lane's \em Strictification Theorem \cite[p.257]{MacLane}, which establishes an equivalence of any SMC with a strict one, enables one to apply the diagrammatic notation to any SMC. 

We will often refer to morphisms of type $\psi: I\rightarrow A$ as \emph{states}, since in a process category they represent the preparation of a system $A$ in a given state. In \textbf{Hilb} such morphisms are in bijection with the vectors of the Hilbert space $A$; in \textbf{Rel} they correspond to subsets of the set $A$. 

The process categories of \emph{quantum-like} theories possess a range of additional structures. 
\begin{definition}
A \emph{dagger category} is a category equipped with a contravariant involutive identity-on-objects functor $(-)^\dag$. A \emph{dagger symmetric monoidal category} ($\dag$-SMC) is a symmetric monoidal category with a dagger functor such that: $(A\otimes B)^\dag = A^\dag \otimes B^\dag$, 
$\lambda_{A}^{-1} = \lambda_{A}^{\dag}$, $\rho_{A}^{-1} = \rho_{A}^{\dag}$, 
$\sigma_{A,B}^{-1} = \sigma_{A,B}^{\dag}$, and $\alpha_{A,B,C}^{-1} = \alpha_{A,B,C}^{\dag}$. 
\end{definition}
\textbf{Hilb} is a $\dag$-SMC with the \emph{adjoint} playing the role of the dagger functor. \textbf{Rel} is a $\dag$-SMC with \emph{relational converse} as the dagger functor. 

\begin{definition}[\cite{KLComClo}]\label{CSdef}
In a SMC $\mathcal{C}$ a \emph{compact structure} on an object $A$ is a tuple 
\[
\{A,A^\ast, \eta_{A}: I\rightarrow A^{\ast}\otimes A, \epsilon_{A}: A\otimes A^{\ast} \rightarrow I\}\,, 
\]
where $A^\ast$ is a dual object to $A$ which may or may not be equal to $A$, and $\eta_A$ and $\epsilon_A$ satisfy the conditions: 
  \begin{equation} \label{CCcomdiag}
  \begin{diagram}
  A&\rTo^{\rho_{A}}&A\otimes I&\rTo^{1_{A}\otimes \eta_{A}}&A\otimes(A^{\ast}\otimes A)\\
  \dTo^{1_{A}}&&&&\dTo_{\alpha_{A,A^{\ast},A}}\\
  A&\lTo^{\lambda_{A}^{-1}}&I\otimes A&\lTo^{\epsilon_{A}\otimes 1_{A}}&(A\otimes A^{\ast})\otimes A\\
  \end{diagram}
  \end{equation}
  and the dual diagram for $A^{\ast}$. 
A \emph{compact closed} category $\mathcal{C}$ is a SMC in which all $A\in \textrm{Ob}(\mathcal{C})$ have compact structures. 
\end{definition}   

\textbf{Hilb} is a compact closed category, where for each object $A$ the \emph{Bell-state} ket $\sum_i \ket{i}\otimes\ket{i}$ (with $\{\ket{i}\}$ a basis for the Hilbert space $A$) familiar from quantum mechanics provides the morphism $\eta_A$, and the corresponding bra acts as $\epsilon_A$. Perhaps unsurprisingly, given this example, Abramsky and Coecke in \cite{CatSem} showed that this structure underlies the capacity of any quantum-like theory to exhibit information processing protocols such as \emph{teleportation} and \emph{entanglement swapping}. Furthermore Abramsky has shown \cite{SamsonNocloning} that any theory whose process category is compact closed will obey a generalised version of the no-cloning theorem. 

For the compact structures of interest in this paper $A^*$ is always equal to $A$ -- see proposition \ref{inducedCSprop} below.\footnote{An analysis of the coherence conditions for these \em self-dual \em compact structures is in \cite{SelingerSelfDual}.} Hence, from here on we won't distinguish these two objects. 

This simplifies the graphical language, which can be extended by introducing special elements to represent the morphisms of the compact structure $\eta_A$ and $\epsilon_A$: 
\begin{equation}\label{unitcounitdiag}
\def\JPicScale{0.4}
\input{unitcounit.pst}
\end{equation}
Equation \ref{CCcomdiag} and its dual are then depicted as: 
\begin{equation}\label{yanking}
\def\JPicScale{0.4}
\input{yanking.pst}
\end{equation}
This extension of the graphical language now renders it completely equivalent to the axioms of a compact closed category; for a detailed discussion we refer to \cite{SelingerSurvey}.
\begin{theorem}
Two morphisms in a compact closed category can be shown to be equal using the axioms of compact closure iff the diagrams corresponding to these morphisms in the graphical language are isotopic. 
\end{theorem}

Compact closure is of particular importance to us since in any compact closed category there will be \emph{map-state duality}: a bijection between the hom-sets $\mathcal{C}(I,A\otimes B)$ and $\mathcal{C}(A, B)$ (in fact both these will further be in bijection with $\mathcal{C}(A\otimes B, I)$):
\begin{equation}\label{abstate2map}
\def\JPicScale{0.8}
\input{abstate2map.pst}
\end{equation}
and 
\begin{equation}\label{abmap2state}
\def\JPicScale{0.8}
\input{abmap2state.pst}
\end{equation}
That this is a bijection follows from equation \ref{yanking}. 
If we have a morphism with larger composite domain and codomain the number of hom-sets in bijection increases dramatically. For example the morphisms of the hom-sets $\mathcal{C}(A_1\otimes\dots\otimes A_m\otimes X, B_1\otimes\dots\otimes B_n)$ and $\mathcal{C}(A_1\otimes\dots\otimes A_m, B_1\otimes\dots\otimes B_n\otimes X)$ are in bijection: explicitly the conversion between morphisms from the two sets can be depicted as: 
\begin{equation}
\def\JPicScale{0.6}
\input{abbendinglines.pst}
\end{equation}
Clearly manoeuvres like this can convert any `input' line into an `output' line, by using the unit and co-unit morphisms to `bend lines around'. 
\begin{definition}
\label{DECdef}
A \emph{diagram equivalence class}\footnote{The terminology is inspired by the fact that the diagrams of all members of a class are essentially the same, only differing in the orientations of their input and output arrows.} (DEC) is a set of morphisms in a compact closed category which can be inter-converted by composition with the units and co-units of the factors of their domains and codomains. 
\end{definition}
The final structure of interest is the \emph{basis structure}, which can be seen as the `dagger-variant' of Carboni and Walters's Frobenius algebras \cite{CarboniWalters}. This structure is discussed at length in \cite{RedGreen}.

\begin{definition}\label{BSdef}
In a $\dag$-SMC a \emph{basis structure} $\Delta$ on an object $A$ is a commutative isometric dagger Frobenius comonoid $(A,\delta:A\rightarrow A\otimes A, \epsilon:A\rightarrow I)$. For more details on this definition see section 4 of \cite{RedGreen} where basis structures are referred to as `observable structures'. We represent the morphisms $\delta$ and $\epsilon$ graphically as: 
\begin{equation}
\def\JPicScale{0.6}
\input{deltaandepsilon.pst}
\end{equation}
\end{definition}

Basis structures are so named because in \textbf{Hilb} they are in bijection with orthonormal bases, via the correspondence: 
\begin{equation}
\delta:\mathcal{H}\rightarrow\mathcal{H}\otimes\mathcal{H}:: |i\rangle \mapsto |i\rangle\otimes |i\rangle \quad \quad \epsilon:\mathcal{H}\rightarrow \mathbb{C}::|i\rangle \mapsto 1
\end{equation}
where $\{|i\rangle\}_{i=1,\dots,n}$ is an orthonormal basis for $\mathcal{H}$. 
This is proved in \cite{CocPavVic}. 

\begin{proposition}[{\rm See e.g.~\cite{RedGreen}}]\label{sepBSprop}
Two basis structures $(A,\delta_A,\epsilon_A)$ and $(B, \delta_B, \epsilon_B)$ induce a third basis structure $(A\otimes B, \delta_{A\otimes B}, \epsilon_{A\otimes B})$, with: 
\begin{equation}
\delta_{A\otimes B}=(1_A\otimes\sigma_{A,B}\otimes 1_B)\circ (\delta_A\otimes \delta_{B}) \qquad\epsilon_{A\otimes {B}}=\epsilon_A\otimes\epsilon_{B}\,,
\end{equation}
\end{proposition}

\begin{proposition}[{\rm See e.g.~\cite{wosums}}]\label{inducedCSprop}
Any basis structure induces a self-dual dagger compact structure, with $A=A^\ast$ and $\eta_A = \delta\circ\epsilon^\dag$. 
\end{proposition}

Basis structures are identified in \cite{RedGreen} as a key structure underlying a variety of quantum-like features, for example the existence of incompatible observables, and information processing tasks such as the quantum fourier transform. 

%%%%%%%%%%%%%%%%%%%%%%%%%
%%%%%%%%%%%%%%%%%%%%%%%%%
\section{The categories Spek and MSpek}
%%%%%%%%%%%%%%%%%%%%%%%%% 
%%%%%%%%%%%%%%%%%%%%%%%%%

%\begin{figure}[htb]
%\begin{center}
%\epsfig{figure=Spek.eps,width=280pt}
%\caption{Category-theoretic graffiti on an Oxford bridge.}
%\end{center}
%\end{figure}

In the section we give the definitions of two key categories. The first definition is a convenient stepping stone to the second: 

\begin{definition}
 \label{defSpek}
The category \textbf{Spek} is a subcategory of \textbf{Rel}, defined inductively, as follows: 
\begin{itemize}
  \item The objects of \textbf{Spek} are the single-element set $\one=\{\ast\}$, the four element set $\four := \{1,2,3,4\}$, and its $n$-fold Cartesian products $\four^n$. 
  \item The morphisms of \textbf{Spek} are all those relations generated by composition, Cartesian product and relational converse from the following generating relations:
  \begin{enumerate}
    \item All permutations $\{\sigma_i : \four\rightarrow\four\}$ of the four element set, represented diagrammatically by: 
\begin{equation}
\def\JPicScale{0.6}
\input{spekperm.pst}
\end{equation}
    There are 24 such permutations and they form a group, $S_{4}$. 
    \item A relation $\delta_{\textbf{Spek}}: \four \rightarrow \four\times \four$ defined by:
    \[\begin{array}{l}
    1\sim \{(1,1), (2,2)\} \\ 2\sim \{(1,2), (2,1)\} \\ 3\sim \{(3,3), (4,4)\} \\ 4\sim \{(3,4), (4,3)\}\,;
    \end{array}\]
    represented diagrammatically by: 
\begin{equation}
\def\JPicScale{0.6}
\input{spekdelta.pst}
\end{equation}
    \item a relation $\epsilon_{\textbf{Spek}}: \four \rightarrow \one$ defined by $\{1,3\}\sim * $
   and represented diagrammatically by: 
\begin{equation}
\def\JPicScale{0.6}
\input{spekepsilon.pst}
\end{equation}
\item the relevant unit, associativity and symmetry natural isomorphisms.
\end{enumerate}
\end{itemize}
\end{definition}

\begin{proposition}{\rm\cite{CE}}\label{Spek:basisstruc}
$(\four,\delta_{\textbf{Spek}},  \epsilon_{\textbf{Spek}})$ is a basis structure in \textbf{FRel} and hence also in \textbf{Spek}. 
\end{proposition}

This category turns out to be the process category for the fragment of the toy theory containing only epistemic states of maximal knowledge. In particular, as discussed in detail in \cite{CE},  the interaction of the basis structure $(\four,\delta_{\textbf{Spek}},  \epsilon_{\textbf{Spek}})$ and the permutations of $S_4$ results in three basis structures, analogous to the $Z$-, the $X$- and the $Y$-bases of a qubit.

Furthermore, as also discussed in \cite{CE}\S4, \textbf{Spek} also contains operations which in quantum theory correspond with projection operators, while these are not included in the toy theory.  Indeed,  compact closure  of \textbf{Spek} implies \em map-state duality\em: operations and bipartite states are in bijective correspondence.\footnote{We see this as an improvement, and in \cite{SpekkensBis} also Spekkens expressed  the desire for theories to have this property, as well as having a dagger-like structure, in the sense that states and effects should be in bijective correspondence.}

The process category for the full toy  theory has a similar definition:  

\begin{definition}
The category \textbf{MSpek} is a sub-category of \textbf{Rel}, with the same objects as \textbf{Spek}. Its morphisms are all those relations generated by composition, Cartesian product and relational converse from the generators of \textbf{Spek} plus an additional generator:
\[ \bot_\textbf{MSpek}: \one\rightarrow\four:: \{\ast\} \sim \{1,2,3,4\} \]
\end{definition}

By construction, both \textbf{Spek} and \textbf{MSpek} inherit dagger symmetric monoidal structure from \textbf{Rel}, with Cartesian product being the monoidal product, and relational converse acting as the dagger functor. 

\begin{proposition}
\textbf{Spek} and \textbf{MSpek} are both compact closed. 
\end{proposition}
\begin{proof}
By propositions \ref{Spek:basisstruc} and \ref{sepBSprop} it follows that each object in these categories has a basis structure, and hence by proposition \ref{inducedCSprop} they both are compact closed.
\end{proof}

The primary task of the remainder of this paper is to demonstrate that \textbf{MSpek} is indeed the process category of Spekkens's toy theory. 

%%%%%%%%%%%%%%%%%%%%%%%%%
%%%%%%%%%%%%%%%%%%%%%%%%%
\section{The general form of the morphisms of Spek and MSpek}\label{aboutdiagsandstates}
%%%%%%%%%%%%%%%%%%%%%%%%%
%%%%%%%%%%%%%%%%%%%%%%%%%

We first make some preliminary observations.

\begin{definition}
A \emph{\textbf{Spek} diagram} is any valid diagram in the graphical language introduced in section \ref{secCategories} which can be formed by linking together the diagrams of the \textbf{Spek} generators, as described in definition \ref{defSpek}. 
\end{definition}

There is clearly a bijection between the possible compositions of \textbf{Spek} generators, and \textbf{Spek} diagrams. A \textbf{Spek} diagram with $m$ inputs and $n$ outputs represents a morphism of type $\four^m\rightarrow\four^n$: a relation between sets $\four^m$ and $\four^n$. The number of relations between two finite sets $A$ and $B$ is clearly finite itself: it is the power set of $A\times B$. Thus the hom-set \textbf{FRel}(A,B) is finite. Since $\textbf{Spek}(\four^m, \four^n) \subseteq \textbf{FRel}(\four^m, \four^n)$ we can be sure that the hom-sets of \textbf{Spek} are finite. On the other hand, there is clearly an infinite number of \textbf{Spek} diagrams which have $m$ inputs and $n$ inputs - we can add more and more internal loops to the diagrams. Thus many diagrams represent the same morphism. However the morphisms of \textbf{Spek} are, by definition, all those relations resulting from arbitrary compositions of the generating relations, i.e. any relation that corresponds to one of the infinity of \textbf{Spek} diagrams. Hence any proof about the form of the morphisms in \textbf{Spek} is going to have to be a result about the relations corresponding to each possible \textbf{Spek} diagram, even though in general many diagrams correspond to a single morphism. 

If we know the relation corresponding to one diagram in one of \textbf{Spek}'s diagram equivalence classes (recall definition \ref{DECdef}), then it is straightforward to determine the relations corresponding to all of the other diagrams. 

\begin{lemma}\label{bendlinelemma}
Given a \textbf{Spek} diagram and corresponding relation:
\begin{equation}
\def\JPicScale{0.6}
\input{relation.pst}
\end{equation}
then the relation corresponding to the following diagram: 
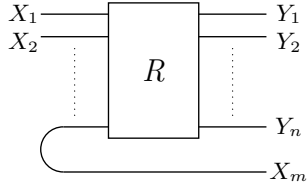
\begin{equation}\label{bendinglines}
\def\JPicScale{0.6}
\input{bendinglines.pst}
\end{equation}
is given by, for all $x_i\in X_i$: 
\begin{equation}
(x_1,\dots,x_{m-1})\sim\left\{(y_1,\dots,y_n,x_m)
\Biggm| 
\begin{array}{c}
y_i\in Y_i\\ 
x_m\in X_m\\ 
(y_1,\dots,y_n)\in R(x_1,\dots,x_m)
\end{array}
\right\}
\end{equation}
where $R(x_1,\dots,x_m)$ is the subset of $Y_1\times\dots\times Y_n$ which is related by $R$ to $(x_1,\dots,x_m)$. 
\end{lemma}

Every diagram equivalence class in \textbf{Spek} has at least one diagram of type $\one\rightarrow\four^n$, representing a state, where we make every external line an output. Relations of this type can be viewed as subsets of the set $\four^n$ and it will be convenient for us to concentrate on characterising these morphisms. Via lemma \ref{bendlinelemma} any results on the general form of states will translate into results on the general form of all morphisms. In what follows we will therefore make no distinction between the inputs and outputs of a \textbf{Spek}-diagram: a diagram with $m$ inputs and $n$ outputs will simply be referred to as a $(m+n)$-legged diagram. 

Our proof will involve building up \textbf{Spek} diagrams by connecting together the generating morphisms. Here we show what various diagram manipulations mean in concrete terms for the corresponding relations. Henceforth, remembering that we only have to consider states, we will assume that the relation corresponding to any $n$-legged diagram is of type $\one\rightarrow\four^n$. 

First we introduce some terminology. The \emph{composite} of an $m$-tuple $(x_1,\dots ,x_m)$ and an $n$-tuple $(y_1,\dots ,y_n)$ is the $(m+n)$-tuple $(x_1,\dots,x_n,y_1,\dots,y_n)$ from $\two^{m+n}$.  By the $i^\textit{th}$\emph{-remnant} of an $n$-tuple we mean the $(n-1)$-tuple obtained by deleting its $i^\textrm{th}$ component. By the $i,j^\textit{th}$\emph{-remnant} of an $n$-tuple (where $i>j$) we mean the $(n-2)$-tuple obtained by deleting the $j^\textrm{th}$ component of its $i^\textit{th}$-remnant (or equivalently, deleting the $(i-1)^\textrm{th}$ component of its $j^\textit{th}$-remnant).

\begin{example}\em\label{relationtree}
Consider linking two diagrams, the first representing the relation $R:\one\rightarrow X_1\times\dots\times X_m$ the second representing the relation $S: \one\rightarrow Y_1\times\dots\times Y_n$ via a permutation $P$, to form a new diagram as shown: 
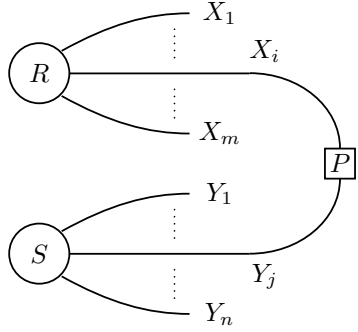
\begin{equation}
\def\JPicScale{0.8}
\input{relationtree.pst} 
\end{equation}
The relation corresponding to this diagram is given by
\begin{equation}
\ast\sim\left\{(x_1,\dots,x_{i-1},x_{i+1},\dots,x_m,y_1,\dots,y_{j-1},y_{j+1},\dots,y_n)
\Biggm| 
\begin{array}{c}
\!\!(x_1,\dots,x_m)\in R(\ast)\!\!\\
\!\!(y_1,\dots,y_n)\in S(\ast)\!\!\\
\!\! x_i = P(y_j)\!\!
\end{array}
\right\}
\end{equation}
Or, in less formal language, for every pair of a tuple from $R$ and a tuple from $S$ obeying the condition $x_i = P(y_j)$, we form composite of the $i^\textrm{th}$ remnant of the tuple from $R$, and the $j^\textrm{th}$ remnant of the tuple from $S$. 
\end{example}

\begin{example}\em\label{relationloop}
Given a diagram representing the relation $R:\one\rightarrow X_1\times\dots\times X_m$, consider forming a new diagram by linking the $i^\textrm{th}$ and $j^\textrm{th}$ legs of the original diagram via a permutation $P$. 
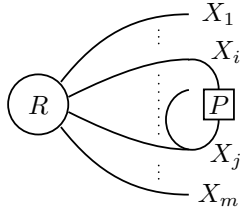
\begin{equation}
\def\JPicScale{0.8}
\input{relationloop.pst}
\end{equation}
The relation corresponding to this diagram is given by: 
\begin{equation}
\ast\sim
\left\{(x_1,\dots,x_{i-1},x_{i+1},\dots,x_{j-1},x_{j+1},\dots,x_m)
\Bigm| 
\begin{array}{c}
(x_1,\dots,x_m)\in R(\ast)\\ 
x_i = P(x_j)
\end{array}
\right\}
\end{equation}
Or, in less formal language, we take the $i,j^\textit{th}$-remnant of every tuple for which $x_i = P(x_j)$. 
\end{example}

\begin{example}\em\label{relationcapping}
Consider linking two diagrams, the first representing the relation $R:\one\rightarrow X_1\times\dots\times X_n$ the second representing the relation $S: \one\rightarrow X_i$ via a permutation $P$, to form a new diagram as shown: 
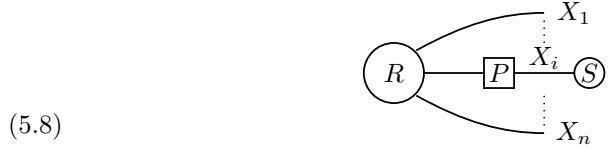
\begin{equation}
\def\JPicScale{0.8}
\input{relationcapping.pst}
\end{equation}
The relation corresponding to this diagram is given by: 
\begin{equation}
\ast\sim 
\left\{
(x_1,\dots,x_{i-1},x_{i+1},\dots,x_n)
\Bigm| 
\begin{array}{c}
(x_1,\dots,x_n)\in R(\ast)\\ 
x_i\in P(S(\ast))
\end{array}
\right\}
\end{equation}
Or, in less formal language, we take the $i^\textrm{th}$ remnant of every tuple for which $x_i\in P(S(\ast))$. 
\end{example}

%%%%%%%%%%%%%%%%%%%%%%%%%
\subsection{Structure of the construction}\label{proofstructureSec}
%%%%%%%%%%%%%%%%%%%%%%%%%

It is convenient to single out a particular sub-group of the $S_4$ permutation sub-group. This consists of the four permutations which don't mix between the sets $\{1,2\}$ and $\{3,4\}$: (1)(2)(3)(4) (the identity), (12)(3)(4), (1)(2)(34) and (12)(34). We term these the \emph{phased} permutations. All other permutations are termed \emph{unphased}. We single out this sub-group because the relations corresponding to \textbf{Spek} diagrams generated from $\delta_{\textbf{Spek}}$, $\epsilon_{\textbf{Spek}}$, and the four phased permutations have a significantly simpler general form. We term such diagrams \emph{phased diagrams}. All other diagrams are termed \emph{unphased diagrams}. A morphism which corresponds to a phased diagram is termed a \emph{phased morphism}, all other morphisms being \emph{unphased morphisms}. 

It is straightforward to see that any unphased \textbf{Spek} diagram can be viewed as a collection of phased sub-diagrams linked together via unphased permutations. We refer to these sub-diagrams as \emph{zones}. Furthermore, note that any permutation in $S_4$ can be written as a product of phased permutations and the unphased permutation (1)(3)(24) (which we denote by $\Sigma$). Since a single phased permutation constitutes a phased zone with two legs, we can in fact view an unphased diagram as a collection of phased zones linked together by the permutation $\Sigma$:
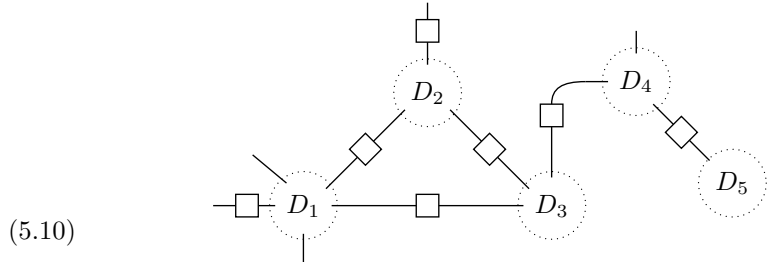
\begin{equation}\label{generalunphasedEq}
\def\JPicScale{0.6}
\input{generalnonphased.pst}
\end{equation}
Here $D_i$ represents a phased sub-diagram and a square box represents an unphased permutations. Note that such a diagram is not necessarily planar, i.e.~it may involve crossing wires.  We distinguish between \emph{external zones} which have external legs (e.g. $D_1$, $D_2$ and $D_5$ in diagram \ref{generalunphasedEq}), and \emph{internal zones} all of whose legs are connected to other legs within the diagram (e.g. $D_3$ and $D_4$).
To the external legs we associate an enumeration, such that legs from the same external zone appear consecutively. 

In the first stage of the proof we determine the general form of the relation corresponding to a \emph{phased diagram}. This stage itself splits into two phases: first we determine the general form of the morphisms of a new category \textbf{HalfSpek}, and then we show how to use this result to prove our main result in \textbf{Spek}. In the second stage we draw on the results of the first to determine the general form of the relation corresponding to \emph{any} \textbf{Spek} diagram. 

%%%%%%%%%%%%%%%%%%%%%%%%%%%%
\subsection{The general form of the morphisms of  HalfSpek}
%%%%%%%%%%%%%%%%%%%%%%%%%

We build up to the full theorem via a simplified case. For this we need a new category. 

\begin{definition}
 \label{defHalfSpek}
The category \textbf{HalfSpek} is a subcategory of \textbf{FRel}. It is defined inductively, as follows: 

\begin{itemize}

  \item The objects of \textbf{HalfSpek} are the single-element set $\one=\{\ast\}$, the two element set $\two := \{0,1\}$, and its $n$-fold Cartesian products $\two^n$. 

  \item The morphisms of \textbf{HalfSpek} are all those relations generated by composition, Cartesian product and relational converse from the following generating relations:
  \begin{enumerate}
    \item All permutations $\{\sigma_i : \two\rightarrow\two\}$ of the two element set.
    There are 2 such permutations, the identity and the operation $\sigma$ which swaps the elements of $\two$.  Together they form the group $Z_{2}$. 
    \item A relation $\delta_{\textbf{Half}}: \two \rightarrow \two\times \two$ defined by:
    \[\begin{array}{ll}
    0\sim \{(0,0), (1,1)\}&\quad 1\sim \{(0,1), (1,0)\}\,;
    \end{array}\]
    \item a relation $\epsilon_{\textbf{Half}}: \two \rightarrow \one::0\sim *$\
  \end{enumerate}
\end{itemize}
\end{definition} 

\begin{proposition}{\rm\cite{CE}}
$(\two,\delta_{\textbf{Half}},  \epsilon_{\textbf{Half}})$ is a basis structure in \textbf{FRel} and hence also in \textbf{HalfSpek}.
\end{proposition}

\begin{remark}\em
The existence of this basis structure came as  a surprise to the authors.  Naively one might think (as many working in the area of categorical quantum mechanics initially did)  that on a set $X$ in \textbf{FRel} there is a single basis structure with $\delta$ given by $x \sim (x,x)$ for all $x\in X$.  The `basis vectors' (or copyable points in the language of \cite{RedGreen}) are then the elements of this set.  But this is not the case.  There are many `non-well-pointed' basis structures such as $(\two,\delta_{\textbf{Half}},  \epsilon_{\textbf{Half}})$ for which the number of copyable points is less than the number of elements of the set.  In related work, Pavlovic has classified all basis structures in \textbf{FRel} and Evans et al.~have identified the pairs of complementary basis structures (in the sense of \cite{RedGreen}) among these \cite{Dusko,Evans}.
\end{remark}

Next we determine the general form of the relations which constitute the morphisms of \textbf{HalfSpek}, to which the considerations made at the beginning of Section \ref{aboutdiagsandstates} also apply.

We say that an element of $\two^n$ has \emph{odd} parity if it has an odd number of `1' elements, and that it has \emph{even} parity if it has an even number of `1' elements.  We will use $P$ to represent a particular parity, odd or even, and $P'$ will represent the opposite parity. Whether an odd-parity $n$-tuple has an odd or even number of `0' elements clearly depends on whether $n$ itself is odd or even. We could have chosen either 0 or 1 to play the role of labelling the parity; we have chosen 1 since it will turn out to be more convenient later on.  

\begin{theorem}\em\label{HalfSpekmorphisms}
The relation in \textbf{HalfSpek} corresponding to an $n$-legged \textbf{HalfSpek}-diagram is a subset of $\two^n$, consisting of all $2^{n-1}$ $n$-tuples of a certain parity, which depending on the diagram may either be even or odd: if the product of all the permutations appearing in the diagram is the identity, then parity is even, and if it is $\sigma$, then the parity is odd. 
\end{theorem}

\begin{proof}
We prove this result by induction on the number of generators $k$ required to construct the diagram. 
Remember that we need only consider those diagrams whose corresponding relations are states. There is just one possible base case ($k=1$), a diagram composed purely of the generator $\epsilon_{\textbf{Half}}^\dag$ for which $n=1$: the corresponding state consists of the single 1-tuple (0), which is indeed the unique 1-tuple of even parity. Now consider a diagram $D$ built from $k$ generators with a corresponding state $\psi$ consisting all $2^{n-1}$ $n$-tuples of parity $P$. It is easily seen that composing $D$ with either $\epsilon_{\textbf{Half}}^\dagger$, $\delta_{\textbf{Half}}$ or $\delta_{\textbf{Half}}^\dagger$ respectively yields a diagram whose corresponding state consists of all $2^{n-2}$ $n-1$-tuples, all $2^{n}$ $n+1$-tuples, and all $2^{n-2}$ $n-1$-tuples of parity $P$; and that composing with $\sigma$ yields all $2^{n-1}$ $n$-tuples  of parity $P'$. Finally consider producing a disconnected diagram by laying the $\epsilon_{\textbf{Half}}^\dag$ diagram alongside $D$: it is easily seen that the corresponding state consists of $2^{n}$ $n+1$-tuples of parity $P$. 
\end{proof}

%%%%%%%%%%%%%%%%%%%%%%%%%
\subsection{The general form of phased morphisms in \textbf{Spek}}
%%%%%%%%%%%%%%%%%%%%%%%%%

We want to apply our results on \textbf{HalfSpek} to the category of real interest, \textbf{Spek}. To do this we first need to digress to discuss some structural features of relations.  The category \textbf{Rel} has another symmetric monoidal structure, namely the \em disjoint union \em or \em direct sum\em, denoted by $\sqcup$. 
Concretely, if we can partition a set $A$ into $m$ subsets $A_i$, then we have $A=\sqcup_i A_i$, and recalling that a relation $R:A\rightarrow B$ is a subset of $A\times B$, we can decompose $R$ into $mn$ components of the form $R_{i,j}:A_i\rightarrow B_j$, such that $R = \bigsqcup_{i,j} R_{i,j}$.  The relations $R_{i,j}:A_i\rightarrow B_j$ are termed the \emph{components} of $R$ \emph{with respect to partitions} $A=\sqcup_i A_i$, $B=\sqcup_j B_j$. In category theoretic terms, this is \em biproduct\em, and there is a \em distributive law \em with respect to the Cartesian product:
\begin{equation}
A\times (\sqcup_i B_i) =  \sqcup_i (A\times B_i) \qquad\quad R\times (\sqcup_i T_i) =  \sqcup_i (R\times T_i)
\end{equation}
for sets $A, B_i$ and relations $R, T_i$. For $A=\sqcup_i A_i$, $B=\sqcup_j B_j$, $C=\sqcup_k C_k$ and $D=\sqcup_l D_l$, and relations $R:A\rightarrow B$, $S:B\rightarrow C$ and $T:C\rightarrow D$, we then have:
\begin{equation}
(S\circ R)_{i,k} = \bigsqcup_j S_{j,k} \circ R_{i,j}\quad
(R\times T)_{i,j,k,l} = T_{k,l}\times R_{i,j}\quad
R^c_{j,i} = (R_{i,j})^c
\end{equation}
For  $A=A_1\sqcup A_2$ by distributivity we have: 
\begin{equation}
A_1^m\sqcup (A_1^{m-1}\times A_2)\sqcup (A_1^{m-2}\times A_2\times A_1)\sqcup\dots\sqcup (A_1^{m-2}\times A_2^2)\sqcup\dots\sqcup A_2^m
\end{equation}
If for $A = A_1\sqcup A_2$ and $R: A^m \rightarrow A^n$ the only non-empty components are $R_1: A_1^m \rightarrow A_1^n$ and $R_2: A_2^m \rightarrow A_2^n$ we call it \emph{parallel}.
Given parallel relations $R: A^m\rightarrow A^n$, $S: A^n\rightarrow A^p$ and $T: A^p\rightarrow A^q$ with respect to the partition $A=A_1\sqcup A_2$, the relations: 
\begin{equation}
S\circ R: A^m\rightarrow A^p\quad\ \ 
T\times R: A^{m+p} \rightarrow A^{n+q}\quad\ \ 
R^c: A^n\rightarrow A^m
\end{equation}
are all also easily seen to be parallel with respect to the same partition of $A$. 

We can use these insights to make a connection between \textbf{HalfSpek} and \textbf{Spek}. 

\begin{proposition}
The generators of the phased morphisms of \textbf{Spek}, i.e. $\delta_\textbf{Spek}$, $\epsilon_\textbf{Spek}$ and the phased permutations on $\four$, are all parallel with respect to the following partition of $\four = \{1,2\}\sqcup\{3,4\}$. We conclude that all phased morphisms of \textbf{Spek} are also parallel with respect to this partition. We refer to the two components of a phased \textbf{Spek} morphism as its $\{1,2\}$-component and $\{3,4\}$-component.  
\end{proposition}

\begin{proposition}\label{HalfSpekgenareSpekgen}
The $\{1,2\}$-components of the generators of the phased morphisms of \textbf{Spek} are simply the generators of \textbf{HalfSpek} with the elements of $\two=\{0,1\}$ re-labelled according to $0\mapsto 1$, $1\mapsto 2$. Similarly the $\{3,4\}$-components of the generators of the phased morphisms of \textbf{Spek} are simply the generators of \textbf{HalfSpek} with the elements of $\two=\{0,1\}$ re-labelled according to $0\mapsto 3$, $1\mapsto 4$.
\end{proposition}

\begin{proposition}\label{SpekDis2HalfSpekD}
A state $\psi\subset\four^n$ corresponding to a phased \textbf{Spek} diagram $D$ is equal to the union of two states $\psi^{12}\subset\{1,2\}$ and $\psi^{34}\subset\{3,4\}$. $\psi^{12}$ and $\psi^{34}$ are obtained by the following procedure. Form a \textbf{HalfSpek} diagram $D^{12}$ by replacing every occurence of $\delta_\textbf{Spek}$ and $\epsilon_\textbf{Spek}$ in $D$ with $\delta_\textbf{Half}$ and $\epsilon_\textbf{Half}$, and replacing every occurence of a permutation with its $\{1,2\}$ component, re-labelled as a $\textbf{HalfSpek}$ permutation as described in proposition \ref{HalfSpekgenareSpekgen}. Form a second \textbf{HalfSpek} diagram $D^{34}$ in the obvious analogous fashion using $\{3,4\}$ components of permutations. $\psi^{12}$ and $\psi^{34}$ are the states corresponding to $D^{12}$ and $D^{34}$, once again under the re-labelling described in proposition \ref{HalfSpekgenareSpekgen}. 
\end{proposition}

Note that $D^{12}$ and $D^{34}$ will appear identical as graphs, both to each other and to $D$, but the labels on some of their permutations will differ. 

From proposition \ref{SpekDis2HalfSpekD} and theorem \ref{HalfSpekmorphisms} now follows:

\begin{theorem}\label{phasedmorphisms}
A phased morphism in \textbf{Spek} of type $\one\rightarrow\four^n$ is a subset of $\four^n$, consisting of $2^n$ $n$-tuples, divided into two classes of equal number: 
\begin{itemize}
  \item The first class consists of tuples of 1s and 2s, all of either odd or even parity. 
  \item The second class consists of tuples of 3s and 4s, again all of either odd or even parity. 
\end{itemize}
\end{theorem}

Note that we are adopting the convention that tuples of the first class have odd parity if they have an odd number of 2s, even parity if they have an even number of 2s. Tuples of the first class have odd parity if they have an odd number of 4s, even parity if they have an even number of 4s. 

%%%%%%%%%%%%%%%%%%%%%%%%%
\subsection{The general form of arbitrary morphisms in \textbf{Spek}}
%%%%%%%%%%%%%%%%%%%%%%%%%

Recall from section \ref{proofstructureSec} that any unphased diagram can be viewed as a collection of phased zones linked together by the permutation $\Sigma$, (1)(3)(24).  %Once a \textbf{Spek}-diagram is factorised in this way, the general form of the morphism corresponding to it is easily stated, and is the subject of the next two theorems: 
We also enumerated the external legs.

\begin{theorem}\label{externalnonphasedTh}
The relation in \textbf{Spek} corresponding to an $n$-legged \textbf{Spek}-diagram with $m$ zones, none of which is internal, is a subset $\psi$ of $\four^n$ with the following form: 
\begin{enumerate}
\item It contains $2^n$ $n$-tuples; each entry corresponding to an external leg.
\item All those entries corresponding to a given zone of the diagram are termed a \emph{zone} of the tuple (whether we are referring to a zone of a diagram, or zone of a tuple should be clear from the context). Each zone of the tuple has a well-defined \emph{type} (components either all 1 or 2, or all 3 or 4) and \emph{parity} (as defined for phased relations). 
\item A \emph{block} $B$ is a subset of $\psi$ such that the $i^\textrm{th}$ zone of all $n$-tuples in $B$ has the same parity and type. The sequence of types and parities of each zone is called the \emph{signature} of the block. The $2^n$ tuples of $\psi$ are partitioned into $2^m$ equally sized blocks each with a unique signature. 
\item Each of the $2^m$ blocks has a different type signature - these exhaust all possible type signatures. 
\item The parity signature of a block is the following simple function of its type signature: 
\begin{equation}\label{typeparityEq}
P_i = \Psi_i(T_i) + \sum_{j\in \textrm{adj}(i)} (T_i + T_j)
\end{equation} 
Here $P_i$ and $T_i$ are Boolean variables representing the parity and type of the $i^\textrm{th}$ zone. We adopt the convention that an odd parity is represented by 0 and an even parity by 1, whilst the type $\{1,2\}$ is represented by 0 and the type $\{3,4\}$ by 1. The set adj($i$) consists of the zones directly adjacent to the $i^\textrm{th}$ zone. $\Psi_i(T_i)$ denotes the parity of the type $T_i$ tuples in the relation corresponding to the $i^\textrm{th}$ zone seen as an independent phased diagram (recall theorem \ref{phasedmorphisms}). 
\end{enumerate}
\end{theorem}

\begin{theorem}\label{internalnonphasedTh}
The relation in \textbf{Spek} corresponding to an $n$-legged \textbf{Spek}-diagram with $m$ zones, of which $m'$ are external is either: 
\begin{itemize}
\item A subset, $\psi$ of $\four^n$ which satisfies conditions (1) and (2) of theorem \ref{externalnonphasedTh} and which is partitioned into $2^{m'}$ blocks. The signatures of these blocks are determined as follows: 
\begin{enumerate}
\item Begin with the state corresponding to the diagram obtained by adding an external leg to every internal zone (i.e. we will have $2^m$ blocks, each with $m$ zones, exhausting all possible type signatures). 
\item Eliminate all blocks whose type signatures do not satisfy the following constraints, one for each internal zone: 
\begin{equation}\label{typeconstraintEq}
\Psi_i(T_i) + \sum_{j\in \textrm{adj}(i)} (T_i + T_j)=0
\end{equation} 
where $i$ is the label of the internal zone, and $j$ labels its adjacent zones. 
\item Finally from each block delete the zones corresponding to internal zones. 
\end{enumerate}
\item \emph{or} it is equal to the empty set, $\emptyset$. This second possibility occurs iff the constraints in equation \ref{typeconstraintEq} are inconsistent. 
\end{itemize}
\end{theorem}

A simple counting argument shows that within each block, every tuple with the correct type and parity signature occurs. Theorems \ref{externalnonphasedTh} and \ref{internalnonphasedTh} thus completely characterise the state corresponding to any \textbf{Spek}-diagram. The input data is the shape of the diagram, which determines \emph{adj(i)}, and the `intrinsic parities' $\Psi_i(T_i)$ of each zone. Futhermore, note from theorem \ref{phasedmorphisms} that, as we would expect, the general form of phased morphisms is a special case of the form described above, with $m=m'=1$. 

\begin{example}\em
We now give an example of theorem \ref{externalnonphasedTh} in action. Consider the following schematic\textbf{Spek}-diagram (circles simply denote zones), which has three zones, all external ($m=3$) and five external legs ($n=5$): 
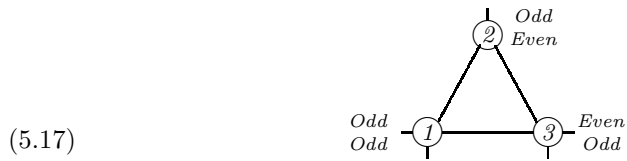
\begin{equation}
\def\JPicScale{0.4}
\input{example.pst}
\end{equation}
where the labels by each zone denote the intrinsic parities of that zone (i.e. the parity of the type-12 tuples and the parity of the type-34 tuples, recall theorem \ref{phasedmorphisms}). 
We conclude from theorem \ref{externalnonphasedTh} that the state corresponding to this diagram will consist of $2^5 = 32$ 5-tuples each with 3 zones, and that these will be partitioned into $2^3=8$ equally sized blocks. Every combination of types appears exactly once amongst these blocks, and the parity signatures are easily determined from equation \ref{typeparityEq} as follows. First we note, from the intrinsic parities that $\Psi_1(T_1) = 0$, $\Psi_2(T_2)=T_2$ and $\Psi_3(T_3)= 1+T_3$. We then see that equation \ref{typeparityEq} here is essentially three equations: 
\begin{equation}
\left\{\begin{array}{ll}
P_1 = & T_2 + T_3\\
P_2 = & T_1 + T_2 + T_3\\
P_3 = & 1 + T_1 + T_2 + T_3
\end{array}\right.\end{equation}
From this we can establish the signatures of the eight blocks: 
\begin{equation}\begin{array}{ll}
(\textrm{Odd},12; \textrm{Odd},12; \textrm{Even},12) & (\textrm{Even},12; \textrm{Even},12; \textrm{Odd},34)\\ (\textrm{Even},12; \textrm{Even},34; \textrm{Odd},12) & (\textrm{Odd},12; \textrm{Odd},34; \textrm{Even},34)\\
(\textrm{Odd},34; \textrm{Even},12; \textrm{Odd},12) & (\textrm{Even},34; \textrm{Odd},12; \textrm{Even},34)\\ (\textrm{Even},34; \textrm{Odd},34; \textrm{Even},12) & (\textrm{Odd},34; \textrm{Even},34; \textrm{Odd},34)
\end{array}\end{equation}
The first zone will have two elements, the second zone one element and the third zone two elements. Within each block every possible collection of tuples consistent with the signature will appear meaning that each block will consist of four tuples. 
\end{example}

\begin{example}\em
We go on to give an example of theorem \ref{internalnonphasedTh}. Consider the following diagram: it is identical to the diagram in the previous example except that the second zone has been internalised. 
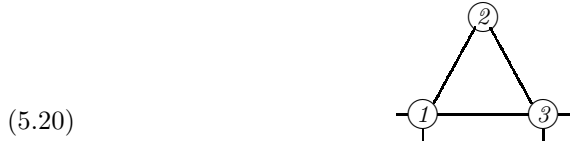
\begin{equation}
\def\JPicScale{0.4}
\input{internalexample.pst}
\end{equation}
This diagram has two external zones $m'=2$ and four external legs $n=4$; thus we expect the corresponding state to consist of $2^4 = 16$ 4-tuples, each with 2 zones, partitioned into $2^2 = 4$ blocks. We now determine the signatures of the blocks applying theorem \ref{internalnonphasedTh}. According to step (1) we begin with the blocks from the previous example. Step (2) requires that we eliminate all blocks whose types do not satisfy the constraint $T_1 + T_2 + T_3 = 0$. This leaves:
\begin{equation}\begin{array}{ll}
(\textrm{Odd},12; \textrm{Odd},12; \textrm{Even},12) & \\ 
 & (\textrm{Odd},12; \textrm{Odd},34; \textrm{Even},34)\\
 & (\textrm{Even},34; \textrm{Odd},12; \textrm{Even},34)\\ 
(\textrm{Even},34; \textrm{Odd},34; \textrm{Even},12) & 
\end{array}\end{equation}
Finally we delete the second zone from each block, leaving: 
\begin{equation}\begin{array}{l}
(\textrm{Odd},12; \textrm{Even},12)\\ 
(\textrm{Odd},12; \textrm{Even},34)\\
(\textrm{Even},34; \textrm{Even},34)\\ 
(\textrm{Even},34; \textrm{Even},12)
\end{array}\end{equation}
\end{example}

The proofs of theorems \ref{externalnonphasedTh} and \ref{internalnonphasedTh} proceed as inductions over the process of building up a diagram by linking together phased zones via the permutation $\Sigma$. In the course of this process it is clearly possible for an external zone to become an internal zone, as its last external leg is linked to some other zone - we refer to this as \emph{internalising} a zone. It turns out that internalisation of a zone complicates the inductive proof. To get around this we do the induction in two stages. In the first stage we build up a diagram identical to the one we are aiming for, \emph{except} that every zone that should be internal is given a single external leg in the following fashion. Suppose we need to link together two zones via a permutation, and this will result in the internalisation of the left hand zone: 
\begin{equation}
\def\JPicScale{0.6}
\input{internalisation1.pst}
\end{equation}
We instead link the left hand zone to the permutation via a $\delta_\textbf{Spek}$ morphism. The $\delta_\textbf{Spek}$ morphism becomes part of the original zone, and provides it with an external leg: 
\begin{equation}
\def\JPicScale{0.6}
\input{internalisation2.pst}
\end{equation}
In the second stage we cap off all these extra external legs with the $\epsilon_\textbf{Spek}$ morphism. Since $\delta_\textbf{Spek}$ and $\epsilon_\textbf{Spek}$ constitute a basis structure the result is the diagram which we are aiming for: 
\begin{equation}
\def\JPicScale{0.6}
\input{internalisation3.pst}
\end{equation}

\subsubsection{Diagrams without internal zones}

We now prove theorem \ref{externalnonphasedTh}. The proof uses an induction over the process of building up a diagram. This building up can be split into two phases: firstly we connect all of the zones which will appear in the final diagram into a `tree-like' structure with no closed loops, secondly we close up any loops necessary to yield the desired diagram. Henceforth we will represent the signature of a tuple with $m$ zones as $(P_1,T_1; \dots ;P_m,T_m)$ where $P_i$ is the parity of the $i^\textrm{th}$ zone, and $T_i$ is its type. Again, if $P$ is a parity, $P'$ indicates the opposite parity, and likewise if $T$ is a type, $T'$ represents the other type. 

\begin{lemma}\label{unphasedtreelemma}
Consider an $n$-legged non-phased diagram $D_1$ with $m$ zones. Suppose the corresponding state $\psi_1\subset\four^n$ satisfies the conditions in theorem \ref{externalnonphasedTh}. Now consider linking the $i^\textrm{th}$ leg of $D_1$ to the $j^\textrm{th}$ leg of an $n'$-legged phased diagram $D_2$ (with corresponding state $\psi_2$), via $\Sigma$, to create an $(n+n'-2)$-legged diagram $D_3$ with $m+1$ external zones. We will assume that the $i^\textrm{th}$ leg of $D_1$ lies within its $k^\textrm{th}$ external zone. The state $\psi_3\subset\four^{n+n'-2}$ corresponding to $D_3$ also satisfies the conditions in theorem \ref{externalnonphasedTh}.
\end{lemma}

\begin{proof}
By the 1-$i^\textrm{th}$-remnants of $\psi_1$ we mean the $i^\textrm{th}$-remnants of those tuples in $\psi_1$ with a 1 in the $i^\textrm{th}$ position. We define the 2-, 3-, and 4-$i^\textrm{th}$-remnants similarly. By proposition \ref{relationtree} the elements of $\psi_3$ comprise all the possible composites of the $x$-$i^\textrm{th}$-remnants of $\psi_1$ and the $\Sigma(x)$-$j^\textrm{th}$-remnants of $\psi_2$, where $x = 1,\dots,4$. It is clear that the zone structure of the tuples of $\psi_1$ is inherited by these composites, and that the $\Sigma(x)$-$j^\textrm{th}$-remnants of $\psi_2$ constitute an additional zone within the composites - we conventionally consider this to be the $(m+1)^\textrm{th}$ zone. Thus the tuples of $\psi_3$ satisfy condition (2). 

Consider a block $B\subset \psi_1$, with signature $(P_1,T_1;\dots;P_k,T_k;\dots;P_{m},T_{m})$, in which for definiteness we assume that $T_k=0$ (the argument runs entirely analogously if $T_k=1$). The composites of the 1-$i^\textrm{th}$-remnants of $B$ and the $1$-$j^\textrm{th}$-remnants of $\psi_2$ all have the same signature, $(P_1,T_1;\dots;P_k,T_k;\dots;P_{m},T_{m};\Psi, T_k)$, and constitute a block $B_1\subset \psi_3$. Likewise the composites of the 2-$i^\textrm{th}$-remnants of $B$ and the $4$-$j^\textrm{th}$-remnants of $\psi_2$ constitute a block $B_2\subset \psi_3$ of signature $(P_1,T_1;\dots;P'_k,T_k;\dots;$ $P_{m},T_{m};\Psi', T'_k)$. Thus each `parent' block in $\psi_1$ gives rise to two `progeny' blocks in $\psi_3$. By hypothesis, each block $B\subset\psi_1$ has a unique type signature, thus the progeny blocks derived from different parent blocks are distinct. Thus $\psi_3$ is partitioned into $2^{m+1}$ blocks, thus satisfying condition (3). It is also clear that if all possible type signatures are represented by the $2^m$ blocks of $\psi_1$ then this is also true for the $2^{m+1}$ blocks of $\psi_3$, and thus that $\psi_3$ satisfies condition (1). 

Note that $B$ consists of $2^{n-m}$ $n$-tuples, and will have $2^{n-m-1}$ 1-$i^\textrm{th}$-remnants and a similar number of 2-$i^\textrm{th}$-remnants, all of which will be distinct. Similarly $\psi_2$ will have $2^{n'-2}$ $1$-$i^\textrm{th}$-remnants and a similar number of 4-$i^\textrm{th}$-remnants, again all distinct. Thus, both $B_1$ and $B_2$ will consist of $2^{n-m-1}.2^{n'-2} = 2^{(n+n'-2)-(m+1)}$ tuples. This holds for all blocks $B\subset \psi_1$, of which there are $2^m$. Thus in total $\psi_3$ consists of $2^{n+n'-2}$ tuples, and so satisfies condition (4). 

Finally we turn to condition (5). Recall from above that a parent block of signature $(P_1,T_1;$ $\dots;P_k,T_k;\dots;P_{m},T_{m})$ yields two progeny blocks, of signatures $(P_1,T_1;\dots;P_k,T_k;\dots;P_{m},T_{m};$ $\Psi, T_k)$ and $(P_1,T_1;\dots;P'_k,T_k;\dots;P_{m},T_{m};\Psi', T'_k)$. Note that those progeny blocks for which the $k^\textrm{th}$ zone and its new adjacent zone have different types exhibit a parity flip on the $k^\textrm{th}$ zone, relative to the parent block, and on the new adjacent zone, relative to its `intrinsic parity'. No such flip occurs if the zones have the same type. Note that the term $T_i + T_j$ is equal to 0 if $T_i=T_j$ and 1 if $T_i\neq T_j$. If condition (2) holds for $\psi_1$ then the correct parity signature for the blocks of $\psi_3$ can be obtained simply by adding the term $T_k + T_{m+1}$ to the $P_k$ and $P_{m+1}$ equations (equation \ref{typeparityEq}). Thus we conclude that condition (5) will hold for $\psi_3$ as well. 
\end{proof}

\begin{lemma}\label{unphasedlooplemma}
Consider an $n$-legged diagram $D$ with $m$ external zones. Suppose the corresponding state $\psi$ satisfies the conditions in theorem \ref{externalnonphasedTh}. Now consider forming a new $(n-2)$-legged diagram $D'$, with corresponding state $\psi'$, by linking the $i^\textrm{th}$ leg of $D$ (in the $k^\textrm{th}$ zone of $D$), to the $j^\textrm{th}$ leg (in the $l^\textrm{th}$ zone), via $\Sigma$. $\psi'$ also satisfies the conditions in theorem \ref{externalnonphasedTh}. 
\end{lemma}

\begin{proof}
By the $x,y$-$i,j^\textrm{th}$\emph{-remnants} of a set of tuples we mean the $i,j^\textrm{th}$-remnants of all those tuples with $x$ in the $i^\textrm{th}$ position and $y$ in the $j^\textrm{th}$ position. From proposition \ref{relationloop}, $\psi'$ consists of the $x,\Sigma(x)$-$i,j^\textrm{th}$-remnants of $\psi$. The zone structure of $\psi$ is clearly inherited by these remnants, thus the tuples of $\psi'$ satisfy condition (2). 

Consider a block $B\subset\psi$ with signature $(P_1,T_1;\dots;P_k,T_k;\dots;P_l,T_l;\dots;P_m,$ $T_m)$. It is straightforward to see that only one quarter of the tuples in this block have $i,j^\textrm{th}$-remnants which are $x,\Sigma(x)$-$i,j^\textrm{th}$-remnants. All the $x,\Sigma(x)$-$i,j^\textrm{th}$-remnants have the same signature: $(P_1,T_1;\dots;P_k,T_k;\dots;P_l,T_l;\dots;P_m,T_m)$ if $T_k = T_l$ and $(P_1,T_1;\dots;P'_k,T_k;\dots;P'_l,T_l;\dots;$ $P_m,T_m)$ if $T_k \neq T_l$. Thus each `parent' block in $\psi$ gives rise to one `progeny' block in $\psi'$, with one quarter as many tuples. The type signatures of the progeny blocks are identical to those of the parent blocks; by hypothesis each parent block had a different type signature and so the progeny blocks deriving from different parent blocks are all distinct. Thus $\psi'$ is partitioned into $2^m$ blocks, each containing $2^{n-m}/4 = 2^{(n-2)-m}$ tuples, and so satisfies conditions (1) and (3). 

Since the type signatures of progeny and parent blocks are identical, if $\psi$ satisfies condition (4), so will $\psi'$. We now turn to condition (5). Closing a loop between the $k^\textrm{th}$ and $l^\textrm{th}$ zones means that they now become adjacent to each other. In the previous paragraph we saw that a progeny block for which the $k^\textrm{th}$ and $l^\textrm{th}$ zones have different types exhibits a parity flip on both these zones. Using similar reasoning as in the previous lemma we conclude that, if condition (2) holds for $\psi$, the correct parity signature for the blocks of $\psi'$ can be obtained simply by adding the term $T_k + T_l$ to the $P_k$ and $P_l$ equations (equation \ref{typeparityEq}). Thus we conclude that condition (5) will hold for $\psi'$ as well. 
\end{proof}

We can now prove \textbf{Theorem \ref{externalnonphasedTh}}:

\begin{proof}
By induction. The base case is a diagram consisting of a single zone, and it is clear from theorem \ref{phasedmorphisms} that this satisfies all the conditions. 
Any other diagram is built up via two inductive steps: linking a new phased zone onto the existing diagram and closing up internal loops within the diagram. Lemmas \ref{unphasedtreelemma} and \ref{unphasedlooplemma} respectively show that if a diagram satisfied the conditions prior to either of these steps, the resulting new diagram will also satisfy the conditions. 
\end{proof}

\subsubsection{Diagrams with internal zones}

We now address the issue of internalising zones. Recall that this step involves capping off external legs with the $\epsilon_\textbf{Spek}$ relation. Throughout this section $D$ will denote a diagram with no internal zones, and $D'$ will denote the diagram obtained by internalising some of $D$'s zones. The corresponding states will be $\psi$ and $\psi'$. 

\begin{proposition}
Suppose we obtain $D'$ by capping the $k^\textrm{th}$ external leg of $D$ with the $\epsilon_\textbf{Spek}$ morphism. From lemma \ref{relationcapping} we conclude that $\psi'$ consists of the 1-$k^\textrm{th}$- and 3-$k^\textrm{th}$-remnants of $\psi$. 
\end{proposition}

Suppose that in going from $D$ to $D'$ we internalised the  $i^\textrm{th}$ zone of $D$. From the proposition above we deduce that each block $B\subset\psi$ for which $P_i=0$ will give rise to one progeny block $B'\subset\psi'$ with the same number of tuples as $B$, while all those blocks for which $P_i=1$ will give rise to no progeny blocks. From this we can conclude that those blocks which do give rise to progeny blocks satisfy a constraint on their type signatures, derived from setting $P_i=0$ in equation \ref{typeparityEq}. Depending on the form of $\Psi_i(T_i)$, and whether the number of zones adjacent to the $i^\textrm{th}$ is odd or even, this constraint takes one of four forms: 
\begin{equation}\label{type0Eq}
\sum_{j\in\textrm{adj}(i)} T_j = \left\{\begin{array}{c} 0 \\ 1 \end{array}\right.
\end{equation}
\begin{equation}\label{type1Eq}
T_i + \sum_{j\in\textrm{adj}(i)} T_j = \left\{\begin{array}{c} 0 \\ 1 \end{array}\right.
\end{equation}
We describe the constraints in (\ref{type0Eq}) as \emph{type-0 constraints}, and those in (\ref{type1Eq}) as \emph{type-1 constraints}. 

Suppose $D$ has $n$ external legs and $m$ external zones. Suppose that in going to $D'$ we internalise $p$ of its zones. \emph{Each internalisation gives rise to a corresponding constraint.} There are now two possibilities: 
\begin{enumerate} 
\item Not all of the constraints are consistent. In this case none of the blocks in $\psi$ satisfy all of the constraints, and none of them will give rise to progeny blocks. Thus $\psi' = \emptyset$. 
\item All $p$ constraints are consistent, and of these $p'$ are linearly independent (this essentially means that $p-p'$ of the constraints can be derived from the remaining $p'$). Each independent constraint reduces the number of blocks which can give rise to progeny by one half. Thus only $2^{m-p'}$ of the blocks in $\psi$ give rise to progeny blocks in $\psi'$, and $\psi'$ can have at most $2^{m-p'}$ blocks - this maximum is attained if all of the progeny blocks are distinct.
\end{enumerate}

\begin{lemma}\label{duplicationLemma}
The following are equivalent: 
\begin{enumerate}
\item The constraints are consistent and there are $p'$ linearly independent constraints. 
\item The $2^{m-p'}$ blocks in $\psi$ which can give rise to progeny blocks in $\psi'$ are partitioned into $2^{m-p}$ sets, each consisting of $2^{p-p'}$ blocks which all yield identical progeny blocks. Thus in total there are $2^{m-p}$ distinct progeny blocks. For brevity we will describe this as $(p-p')$-\emph{fold duplication} of progeny blocks. 
\end{enumerate}
\end{lemma}

The proof of this lemma requires a number of preliminary definitions. 

\begin{definition}
The \emph{IZ-set} is the set of zones which are internalised in going from $D$ to $D'$. 
The \emph{non-internalised adjacent zones (nIAZs)} of an internalised zone are all the zones adjacent to it which are not themselves members of the IZ-set. 
An \emph{adjacency closure set (ACS)} is a subset of the IZ-set with the minimal number of elements such that the disjoint union of the nIAZs of each element contains each nIAZ an even number of times. 
\end{definition}

\begin{example}\em
Consider the following seven zone diagram (external legs are suppressed for clarity). The filled-in zones are those which we internalise. 
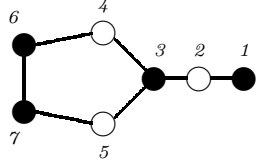
\begin{equation}
\def\JPicScale{0.6}
\input{ACSexample.pst}
\end{equation}
The IZ-set is $\{1,3,6,7\}$. The nIAZs for 1 are $\{2\}$ for 3 are $\{4,5\}$, for 6 are $\{4\}$ and for 7 are $\{5\}$. Zones 3, 6 and 7 together constitute an ACS. Zone 1 is not part of any ACS. 
\end{example}

\begin{definition}
Given a set $S$ of zones in $D$, and a block $B\subset\psi$, the \emph{$S$-mirror of $B$}, $B_S$ is the block with the same type signature as $B$ except on the zones in $S$, where the types are opposite. 
\end{definition}

\begin{proposition}\label{dupmirrorProp}
Suppose $D$ has an ACS $R$. Now, so long as the blocks $B,B_R\subset\psi$ (i.e. a block and its $R$-mirror) both yield progeny blocks in $\psi'$, these progeny blocks will be identical. Conversely, if any two blocks $B,B'\subset\psi$ yield identical progeny blocks in $\psi'$, they must be mirrored with respect to some ACS in $D$. 
\end{proposition}

\begin{proof}
For two blocks in $\psi$ to give identical progeny blocks in $\psi'$ they must have identical type and parity on every zone which is not internalised. Note from equation \ref{typeparityEq} that two blocks in $\psi$ which are type mirrored on a single zone will otherwise differ only in parities on all the zones adjacent to this zone. Now consider two blocks which are type-mirrored on a set of zones $R$: in the case where $R$ constitutes an ACS all the parity flips predicted by equation \ref{typeparityEq} cancel one another out on the zones which will still be visible in $\psi'$.
\end{proof}

\begin{proposition}\label{mirrorconstraintProp}
Suppose $D$ has an ACS $R$ whose member zones have corresponding constraints which satisfy the following condition: those zones with an odd number of adjacent zones within the IZ-set have type-1 constraints, while those zones with an even number of adjacent zones within the IZ-set have type-0 constraints. Then if a block $B\subset\psi$ satisfies the constraints and gives rise to progeny blocks in $\psi'$, so does its $R$-mirror. Conversely, if a block $B\subset\psi$ and its mirror with respect to some ACS $R$ in $D$ both give rise to progeny blocks, the constraints corresponding to the zones of $R$ must satisfy the condition above. 
\end{proposition}

\begin{proof}
If a constraint contains an even number of terms relating to zones from a set $R$ then given a block $B\subset\psi$ either (i) both $B$ and $B_R$ satisfy the constraint (ii) neither $B$ nor $B_R$ satisfy the constraint. Conversely, if both $B$ and $B_R$ satisfy a constraint, it must contain an even number of terms relating to zones from $R$. 
\end{proof}

%\bR[PROBLEM 1] I've realised that this argument doesn't quite work. The problem comes when we have an ACS with a zone which is adjacent to some other zone which \emph{is} internalised, but which \emph{isn't} in the ACS.  See me for an example. When I worked out this result I neglected such cases, but they screw the proof up. \e

\begin{proposition}\label{constraintlindepProp}
Suppose $D$ has an ACS $R$ whose member zones have corresponding constraints which satisfy the condition in the previous proposition. Then the constraints together form a linearly dependent set. The converse is also true. 
\end{proposition}

\begin{proof} 
If $R$ is an ACS and the condition on constraints is satisfied then each term appears in the constraints an even number of times altogether. Summing all the constraints together then results in all the terms cancelling out, yielding the single equation 0=0. This is a necessary and sufficient condition for the constraints to be linearly dependent. 
\end{proof}

We can now prove lemma \ref{duplicationLemma}.

\begin{proof}
From propositions \ref{dupmirrorProp} and \ref{mirrorconstraintProp} we conclude that for $n$-fold duplication to take place the IZ-set must contain $n$ ACSs and that the constraints must satisfy the condition in proposition \ref{mirrorconstraintProp}. The converse is also clearly true. 
Every linearly dependent set amongst the constraints corresponding to the internalised zones reduces the total number of linearly independent constraints by one. From proposition \ref{constraintlindepProp} we conclude that for there to be $n$ linearly dependent sets the IZ-set must contain $n$ ACSs and that the constraints must satisfy the condition in proposition \ref{mirrorconstraintProp}. The converse is also clearly true. 
Thus we conclude that both statements in the lemma are equivalent to a third statement: the IZ-set contains  $p-p'$ ACSs, and the constraints corresponding to the internalised zones satisfy the condition of proposition \ref{mirrorconstraintProp}. 
\end{proof}

\begin{corollary}
Given a diagram $D$ without internal zones which satisfies the conditions in theorem \ref{externalnonphasedTh}, a diagram $D'$ with internal zones formed by capping off external legs of $D$ with $\epsilon_\textbf{Spek}$ morphisms will satisfy the conditions in theorem \ref{internalnonphasedTh}. 
\end{corollary}

\textbf{Theorem \ref{internalnonphasedTh}} follows as a straightforward corollary. 

%%%%%%%%%%%%%%%%%%%%%%%%%
\subsection{The general form of the morphisms of  MSpek}
%%%%%%%%%%%%%%%%%%%%%%%%%

\begin{theorem}\em\label{MSpeknumbers}
All \textbf{MSpek} morphisms of type $\one\rightarrow\four^n$ are subsets of $\four^n$ containing $2^n, 2^{n+1},\dots,2^{2n-1}$ or $2^{2n}$ $n$-tuples. 
\end{theorem}

\begin{proof}
Any \textbf{MSpek} diagram $D'$ can be obtained from a \textbf{Spek} diagram $D$ simply by capping one or more legs of $D$ with the morphism $\bot_\textbf{MSpek}$. 
Suppose we obtain $D'$ by capping a single external leg of $D$ (the $i^\textrm{th}$, say) with $\bot_\textbf{MSpek}$. The state $\psi'$ corresponding to $D'$ consists of the $i^\textrm{th}$-remnants of $\psi$, the state corresponding to $D$. Suppose $D$ has $n$ external legs: since $D$ is a \textbf{Spek}-diagram $\psi$ consists of $2^n$ tuples. Then, unless some of the $i^\textrm{th}$-remnants of $\psi$ are identical, $\psi'$ will also consist of $2^n$ tuples, despite only having $n' = n-1$ external legs. Furthermore it is
 clear that either all the $i^\textrm{th}$-remnants of $\psi$ are distinct, or $\psi$ is partitioned into pairs, the elements of which yield identical $i^\textrm{th}$-remnants, meaning that the addition of each $\bot_\textbf{MSpek}$ cap either halves the number of tuples or leaves it unchanged. 
\end{proof}

%%%%%%%%%%%%%%%%%%%%%%%%%
%%%%%%%%%%%%%%%%%%%%%%%%%
\section{MSpek is the process category for the toy theory}
%%%%%%%%%%%%%%%%%%%%%%%%%
%%%%%%%%%%%%%%%%%%%%%%%%%

We know that the epistemic states of the toy theory are subsets of the sets $\four^n$, and that the transformations on these states are relations between these sets. Thus we can see immediately that the toy theory's process category must be some sub-category of \textbf{FRel}, restricted to the objects $\four^n$. Furthermore we know that it cannot be the full sub-category restricted to these objects, since some subsets of $\four^n$ clearly violate the knowledge balance principle. We will show now that  (a strong candidate for)\footnote{Given that the toy theory is in fact not unambiguously defined for more than three systems, there may be other extensions too.  Ours is the minimal extension given compositional closure.}  the process category for the toy theory in its entirety is \textbf{MSpek}, while if we restrict the toy theory to states of maximal knowledge (consistent with the knowledge balance principle), the process category is \textbf{Spek}. 

\begin{proposition}
The morphisms of the process category of the toy theory are closed under composition, Cartesian product and relational converse. 
\end{proposition}

\begin{proof}
There is no feature of the toy theory which would put any restrictions on which operations could be composed, so we expect the states and transformations to be closed under composition. Since the Cartesian product is used by the toy theory to represent composite systems we also expect the states and transformations to be closed under Cartesian product.

Every epistemic state corresponds to an outcome for at least one measurement (measurements correspond to asking as many questions as possible from canonical sets, epistemic states correspond to the answers). Recalling the discussion of measurement in section 2, we see that given a state $\psi\subset\four^n$ the disturbance resulting from the corresponding measurement outcome can be decomposed as $\psi\circ\psi^\dag$, where $\psi^\dag$ is the relational converse of $\psi$. Thus we expect the relational converse of each state also to feature in the physical category of the theory. The toy theory state corresponding to the subset $\Psi_\textbf{Spek} = \{(1,1),(2,2),(3,3),(4,4)\}\subset\four\times\four$ along with its relational converse are then easily seen to constitute a compact structure on $\four$. We thus have map-state duality, and it is straightforward then to show that if states are closed under relational converse, so is any morphism in the physical category. 
\end{proof}

Note that this point sharpens our discussion about the consistency of the toy theory, in  section 2. If the states and transformations which Spekkens has derived for up to three systems, under the operations of composition, Cartesian product and relational converse, yield states which violate the knowledge balance principle, then the theory as presented is inconsistent. 

\begin{proposition}
All of the generating morphisms of \textbf{MSpek} are states or transformations of the toy theory, or can be derived from them by composition, Cartesian product or relational converse. 
\end{proposition}

\begin{proof}
The only generator for which this is less than obvious is $\delta_\textbf{Spek}$. This is formed by composing Spekkens's GHZ-like state (see section V of \cite{Spekkens}) with the relational converse of the state $\Psi_\textbf{Spek}$ defined in the proof above. 
\end{proof}

\begin{proposition}
All of the states and transformations derived by Spekkens in his original paper \cite{Spekkens} are morphisms of \textbf{MSpek}. When we restrict to states of maximal knowledge all of the states and transformations are morphisms of \textbf{Spek}. 
\end{proposition}

\begin{proof}
By inspection of \cite{Spekkens}.
\end{proof}

\begin{corollary}\label{Spekisclosure}
\textbf{MSpek} is the minimal closure under composition, Cartesian product and relational converse of the states and transformations described in \cite{Spekkens}. \textbf{Spek} is the minimal closure under these operations of the states of \cite{Spekkens} corresponding to maximal knowledge and the transformations which preserve them. 
\end{corollary} 

\begin{proposition} \label{SpekMSpekareKBP}
All states $\psi:\one\rightarrow\four^n$ of \textbf{MSpek} and \textbf{Spek} satisfy the knowledge balance principle on the system corresponding to $\four^n$ viewed as one complete system. All those of \textbf{Spek} satisfy the principle maximally. 
\end{proposition} 

\begin{proof}
Recall that the knowledge balance principle requires that we can know the answer to at most half of a canonical question set. A system with $n$ elementary components has $2^{2n}$ ontic states. A canonical set for such a system consists of $2n$ questions, each answer to a question halving the number of possibilities for the ontic state. Thus, we know the answer to $m$ such questions ($m = 0,\dots,n$), iff our epistemic state is a subset of $\four^n$ with $2^{2n-m}$ elements. We conclude from theorem \ref{MSpeknumbers} that all states of \textbf{MSpek} satisfy the knowledge balance principle on the system as a whole. We conclude  that all states of \textbf{Spek} correspond to the maximum knowledge about the system as a whole consistent with the knowledge balance principle. 
\end{proof}

\begin{proposition}\label{subsysareKBPtoo}
All states $\psi:\one\rightarrow\four^n$ of \textbf{MSpek} and \textbf{Spek} satisfy the knowledge balance principle on every subsystem of the system corresponding to $\four^n$. 
\end{proposition} 
 
\begin{proof}
Given an epistemic state $\psi\subset\four^n$ of a composite system with $n$ elementary components, the `marginal' state on some subsystem is obtained from $\psi$ by deleting from the tuples of $\psi$ the components corresponding to the elementary systems which are not part of the subsystem of interest. Suppose this epistemic state corresponds to a \textbf{Spek} or \textbf{MSpek} diagram, $D$. The elementary systems which are not part of the subsystem correspond to a certain collection of external legs of $D$, and, by lemma \ref{relationcapping}, if we cap these with the \textbf{MSpek} generator $\bot_\textbf{MSpek}$, the effect on the state $\psi$ is exactly as just described. Composing a \textbf{Spek} or \textbf{MSpek} morphism with $\bot_\textbf{MSpek}$ yields some morphism of \textbf{MSpek}, which by proposition \ref{SpekMSpekareKBP} satisfies the knowledge balance principle. 
\end{proof}

From corollary \ref{Spekisclosure} and propositions \ref{SpekMSpekareKBP} and \ref{subsysareKBPtoo}  we reach two key conclusions: 

\begin{itemize}
\item The states and transformations derived by Spekkens in \cite{Spekkens} for systems of up to three components are all consistent with the knowledge balance principle. 
\item The process category of the toy theory must, at least, contain all of the morphisms of \textbf{MSpek}. 
\end{itemize}

The second conclusion begs the question, could \textbf{MSpek} be a strict sub-category of the process category of the toy theory i.e. could the toy theory contain operations not contained in \textbf{MSpek}? It is
 difficult to answer this question, since, as discussed at the end of section 2 it is not clear what the rigorous definition of the toy theory is, or whether there is an unambiguous way to extend it beyond three systems. Certainly, \textbf{MSpek} is the process category of a theory which coincides with Spekkens's theory up to the case of three qubits, and whose states and transformations are bound to satisfy the three rules of section 2 (the first two rules by propositions \ref{SpekMSpekareKBP} and \ref{subsysareKBPtoo}, and the third simply by its definition as the closure under composition of a set of generators). It is in this sense that we earlier remarked that \textbf{MSpek} is a strong candidate for the process category of the toy theory. 

%%%%%%%%%%%%%%%%%%%%%%%%%
%%%%%%%%%%%%%%%%%%%%%%%%%
\section{Conclusion and outlook}
%%%%%%%%%%%%%%%%%%%%%%%%%
%%%%%%%%%%%%%%%%%%%%%%%%%

We achieved our goal stated in the abstract, that is, to provide a rigorous mathematical description of Spekkens' toy theory, which proves its consistency.  This was established both in terms of generators for a dagger symmetric monoidal subcategory of \textbf{FRel}, consisting of symmetries for the elementary system and a basis structure (and nothing more!), as well as in terms of an explicit description of these relations as in Theorems \ref{internalnonphasedTh}, \ref{externalnonphasedTh} and \ref{MSpeknumbers}.

This description meanwhile already has proved to be of great use, for example, in pinpointing what the essential structural difference is between the toy theory and the relevant fragment of quantum theory.  In joint work with Spekkens in \cite{CES} we showed that the key difference between the toy theory and relevant fragment of quantum theory is the \em phase group\em, a group that by pure abstract nonsense can be attributed to each basis structure.  In the case of the toy theory this phase group is $Z_2\times Z_2$ while in the case of the relevant fragment of quantum theory it is $Z_4$.  One can then show that it is this difference that causes the toy theory to be local, while the relevant fragment of quantum theory is non-local.  

In this context, one may wonder wether there is a general categorical construction which would turn a `local theory' like \textbf{Spek} into a non-local one.  We also expect that the construction in this paper can be fairly straightforwardly extended beyond qubit  theories, for example to qutrits \cite{SpekkensTris}. 

%%%%%%%%%%%%%%%%%%%%%%%%%
%%%%%%%%%%%%%%%%%%%%%%%%%
\bibliographystyle{amsplain}
%\begin{thebibliography}{99}
%%%%%%%%%%%%%%%%%%%%%%%%%
%%%%%%%%%%%%%%%%%%%%%%%%%

\end{document}